\documentclass[10pt,conference]{IEEEtran}

\usepackage{amsmath,amssymb,amsthm}
\usepackage{graphicx,color}
\usepackage{multirow}
\usepackage{bm}

\usepackage[usenames,dvipsnames]{pstricks}
\usepackage{epsfig}
\usepackage{pst-grad} 
\usepackage{pst-plot} 
\usepackage{algorithmic}
\usepackage{algorithm}
\newtheorem{theorem}{Theorem}

\usepackage{amsmath}

\title{Joint Routing, Scheduling and Power Control Providing Hard Deadline in Wireless Multihop Networks{\fontsize{0.75cm}{1}\footnotemark}}
\author{\IEEEauthorblockN{Satya Kumar V, and Vinod Sharma}
\IEEEauthorblockA{Department of Electrical Communications Engineering,\\Indian Institute of Science, Bangalore \\
Email: satyakumar@ece.iisc.ernet.in, vinod@ece.iisc.ernet.in}
}
\date{}
\begin{document}
\maketitle
\begin{abstract}
We consider  optimal/efficient 
power allocation policies in a single/multihop wireless
network in the presence of hard end-to-end deadline delay constraints on 
the transmitted packets.
Such constraints can be useful for real time voice and video. 
Power is consumed in only  transmission of the data.
We consider the case when the 
power used in transmission is a convex function of the data transmitted. 
We  develop a computationally efficient online algorithm,
which minimizes the average power for the
single hop. 
We model this problem as dynamic program (DP)
and obtain the optimal solution. 
Next, we generalize it to the multiuser, multihop scenario when there are multiple real time streams
with different hard deadline constraints.
\end{abstract}

\section{Introduction}
The Telecommunication field  is growing at a tremendous pace across the globe
in the last few decades. 
The number of mobile users and mobile internet applications 
are growing exponentially \cite{exponential}. 
Users are looking for various services such as voice calls, video calls, data 
and internet of things (IoT) applications.
One of the  reasons for such a high mobile user growth   is that
people  are using mobile phones for their business purposes as well.
The adoption of mobile technology by citizens of a
country positively affects both the income 
of its citizens as well as the gross domestic product (GDP) of the country.
At the same time growing
carbon footprint of Telecommunication industry has been a
cause of concern and green communications has been the goal
of next generation cellular systems (\cite{angel}, \cite{satya}). Thus, this
paper addresses the question of providing Quality of Service
(QoS) to real time applications while minimizing the transmit
power.

 A specific quality of service may be desired or required for certain types of network 
traffic \cite{walrand}.
For example, hard deadline may be needed for  streaming media,  internet protocol television (IPTV),
   Voice over IP (VoIP), 
 video-conferencing,   safety-critical applications such as remote surgery, and real-time control of machinery.
 However, for TCP file transfers and web browsing, a lower bound on the mean rate provided may be an appropriate
 QoS.

 In the following we survey the related literature. 
 \cite{gold} provided the power allocation policy 
for a single fading link which optimizes the rate. They ignored higher
layer performance measures like queueing delay.
Energy efficient scheduling under mean delay constraint was first addressed in 
 (\cite{rb}, \cite{rarg}).  
Cross layer scheduling algorithms which stabilized
a communication network were considered in
(\cite{a001}, \cite{a003}, \cite{a002},  
\cite{a004}, \cite{a005}). 
Algorithms which minimize mean delay were designed in \cite{tara}.

  \cite{goyal}  considered the problem of
minimizing average delay under
  average power constraint, proved existence of an optimal policy
and obtained structural results for the optimal policy. 
Near-optimal closed-form solution is obtained in \cite{satya} that minimizes the average queue 
length under average power constraint when the rate is a linear function of power.
In \cite{satya12} considered the problem of minimizing the average power under average queue constraint,
and show the existence of an optimal policy.
\cite{nitin} implemented an online algorithm 
by modifying the value iteration equation that minimizes the average
power under an average delay constraint for a single user.
Existence of stationary optimal polices for the constrained average-cost
Markov decision processes
 was shown by setting up the problem as a constrained Markov decision process  (C-MDP) in 
(\cite{b002}, \cite{b003}, \cite{b001}).
Using techniques in Markov decision process  (MDP), structural properties of the optimal policy were obtained in 
\cite{c001}.

\cite{nel} proposed a new innovative algorithm which optimizes power while satisfying an upper
bound on the sum of the average queue lengths of multiple users by dropping packets intelligently. 
\cite{wang} proposed a  suboptimal policy which minimizes the average power under 
average queue constraint when there is an upper bound on packet loss also.
\cite{vineet} obtained asymptotic lower bounds for average queue length and average power consumption.
\cite{starneely} presented an algorithm for a
multiuser and multi channel scenario subject to an upper bound 
on the sum of the average queue length.
To meet the constraints the algorithm needs to learn the system parameters.

 {Energy efficient schemes are proposed when there is a 
hard deadline constraint over a wireless fading channel
in (\cite{wchen}, \cite{f004}, \cite{f003}, \cite{f002}, \cite{f001}). 
\cite{modi} presented a policy which minimizes the energy for sending a 
fixed number of packets
under given hard deadline  delay constraints. 
These  works studied the scenario when there is a energy harvesting system.
 \cite{f005} obtained a closed form optimal average power  solution when the hard deadline
is two and proposed a sub-optimal solution when the hard deadline constraint is more than 
two. \cite {wchen}  proposed an energy efficient scheduler with individual packet hard delay constraints.
In \cite{satyakanpur} obtains optimal solution that minimizes the average power under hard deadline 
constraint, when the rate is a linear function of power.

{Multihop QoS  problem can be solved in either a distributed or a centralized manner.  
Work on joint 
routing, scheduling and power control was 
provided in \cite{Cao_etal} which maximizes a utility function under average power constraint.
As this problem is intractable they provided a heuristic sub-optimal algorithm.
  \cite{vs_mh} considered the problem of
ensuring a fair utilization of network resources by jointly optimizing power 
control, routing, and scheduling  and obtained an efficient sub-optimal solution. 
 \cite{harish} extended the solution in
\cite{Cao_etal} to a multihop network where 
different nodes have multiple antennas and presented an efficient algorithm
for providing max-min fairness.}

\cite{hungneely} uses quadratic
Lyapunov functions to provide novel back-pressure algorithms. 
In \cite{neelyfrnds}, using the above approach upper bounds on average delay are presented. 
These back-pressure algorithms provide stability of the network if the load is within the capacity region.
But under high load, the end-to-end delay will be large and may violate any mean delay constraints.

A distributed scheme for joint power control, 
 scheduling and routing is proposed in  (\cite{g004}, \cite{g005},  \cite{g003}) for wireless networks that
guarantees the attainment of a certain fraction of the capacity region under
the signal to interference ratio (SINR) model. 
 Gossip algorithms are surveyed  in \cite{g007}. 
{In \cite{ashok}, authors  proposed a distributed algorithm  which uses a
 randomized approach to make provision for QoSs such as end-to-end mean delay/hard deadline
 delay.  \cite{prkumar} proposes a distributed scheme,  based on
MDP, to ensure end-to-end hard deadline constraints.

{Our contribution}:
We consider a multihop wireless network where the links experience fading.
 We have  obtained closed-form solutions, 
that minimize the average power 
when there is a hard deadline  constraint for a single user, single hop scenario.
The problem is formulated as a dynamic program (DP). By observing the solution structure of the DP, 
we  obtain an elegant closed-form expression for the optimal policy.
Later on, we  extend our single user, single hop hard deadline  results
to the  single/multiuser, multihop scenario, when the rate is a logarithmic function of power
and obtain computationally efficient algorithms when every user has its own end-to-end
hard deadline constraint. We obtain efficient routing, scheduling and power control
to provide end-to-end hard deadline for the users.

Our model is close to that of \cite{f005}. But \cite{f005} only considers a single user transmitting
over one hop. They obtain optimal policy for a deadline of 2 slots only and obtain heuristics for higher 
deadline. In \cite{prkumar} multiuser and multihop hard deadline is considered. But there is no fading on the 
transmission links, no power control and the arrival processes are deterministic. In our case we 
consider channels with fading, arrival processes are random and optiomal power scheme is obtained. In 
\cite{ashok} the above mentioned limitations of \cite{prkumar} are not there but it provides random
routing and does not optimize power. 

The paper is organized as follows. In Section \ref{s4_sm}
we describe the system model. 
In Section \ref{concave-case1} we consider the system when
data packets arrive to the queue periodically after $M$ slot intervals.  
 Section \ref{concave-case2} considers the case when the data packets arrive in every slot of the frame
and should be served by the end of the frame. 
In Section \ref{concave-case4} we extend the results of single user, single hop to single user, multihop
scenario.
In Section \ref{concave-case5} we generalize the results to the case when 
multiple, real time streams arrive, each with its own hard delay 
constraint.
Finally, Section \ref{s4_conclusions} concludes the chapter.
\section{System Model} \label{s4_sm}
Initially we consider a single user, single hop system. This will be later on
generalized to a multiuser, multihop system.

We consider a discrete time queue, where time is divided 
into slots of  duration one unit. Let $A_k$ nats arrive in 
the queue at the beginning of slot $k$ and they are stored in an 
infinite buffer. These should be transmitted within  next $M$ slots
(i.e., in time $[k,(k+M)))$, where $M<\infty$ is a 
positive integer (see Sections III and IV
for  arrival processes considered there).  
In a practical system, this corresponding to the case that multiple 
number of packets arrive and at the time of transmission these can be fragmented arbitrarly.
In a wireless system, this is a common practice. 
  We assume that the channel gain $H_k$ is
constant over the duration of slot $k$ and take value in a finite set. 
If the channel gains take continuous values, e.g., Rayleigh distributed,
these can be approximated well by quantization by taking $L$ large enough. 
The channel gain $H_k$ is known to the transmitter and receiver at 
time $k$. We assume $\{ H_k,\ k \geq 0\}$ is independent, identically distributed (\emph{iid}).
Let $R_k$ be the number
of nats transmitted in slot $k$.

Let $q_k$ denote the number of nats in the buffer at time $k$. 
Then, the queue evolves as, 
\begin{equation}
q_{k+1}= (q_{k}+ A_{k}- R_{k})^{+},\ k \geq 0, \label{basic2}
\end{equation}
where ${(x)}^+ = \max (0,x)$.
In slot $k$, the power required $P_k$ to transmit $R_k$ (nats) is given by 
 Shannon formula,
\begin{equation}
 R_k= \ln(1+ \frac{ \gamma P_k H_k}{\sigma^2}) , \label{first}
\end{equation}
where $\ln$ denotes $\log$ with base $e$,  $\sigma^2$ is the noise variance and $\gamma$ depends on modulation and coding used. 
Our objective is to minimize,
\begin{equation}
 \limsup_{n \rightarrow \infty} \frac{1}{n} \sum_{k=1}^n E[P_k], \label{main}
\end{equation}
when there is an individual delay constraint on each packet, i.e., $A_k$ should 
be transmitted by time $(k+M), \forall k$.

In Section III we obtain the optimal policies where the time axis is divided
into  frames of size $M$ time units (first frame is $[1, M]$), where $M$ is the deadline of each packet. 
The arrivals come in the beginning of a frame. In Section IV we allow the packets to arrive in every slot but 
the packets arriving in time $[1, M]$, need to be transmitted by time $M$. The policies obtained in Sections III and IV
will be used in later sections to deveop routing, scheduling and power control policies in multihop wireless
networks providing end-to-end hard deadlines.
\section{Arrivals in Beginning of Frame} \label{concave-case1}
  For simplicity, in Shannon formula (\ref{first}), we assume $\sigma^2=1$, and $\gamma=1$.
  Our optimal policies obtained below can be generalized easily.
 We assume that $A_k$ nats arrive at time $kM+1$, $k=0,1,2,...$ and need to be transmitted by time $(k+1)M$.
 No other arrivals come in the mean time. We assume $\{A_{kM}, k \geq 0\}$ is an $iid$ sequence. 
 We provide an algorithm for this setup. It uses Dynamic Programming \cite{putterman}. The intervals $[kM+1, (k+1)M]$
 will be called frames of size $M$.
  The following theorem will be used to obtain the optimal
 algorithm below.
 
 \begin{theorem}
    {    The optimal average power consumption }
 \begin{eqnarray} \label{convexhardstar}
 E_H[W_M(A_1,H)]  =  M e^{\frac{A_1}{M}} \left( \displaystyle{\prod_{j=1}^M} E \left[ \frac{1}{H^{\frac{1}{j}}} \right]^{\frac{j}{M}} \right) \notag \\
 -M E\left[ \frac{1}{H} \right], 
 \end{eqnarray}
for $M\geq 1$.

 \end{theorem}
  \begin{proof}
  {We define the set of all feasible policies 
 $ S_M({A_1})=\{ R=(R_1,...,R_{M}):  \sum_{k=1}^{M} R_k=A_1 ,R_k \geq 0, \forall k \} $. 
  The expected cost for choosing policy $R$ is}
 \begin{equation}
  \displaystyle{ W_M(A_1,R,H)=E_{H}\left[ \sum_{k=1}^{M} f(R_k,q_k,H_k)  \right ], }
 \end{equation}
 {where $f(R_k,q_k,H_k)$,  the power consumed in slot $k$ by transmitting
$R_k$ nats when $q_k$ is the queue length in slot $k$, is}
\begin{equation}
 \displaystyle{ f(R_k,q_k,H_k)= \left( \frac{ e^{R_k}-1}{H_k} \right) \text{ for } {\{R_k \leq q_k\}}, \forall k}.
\end{equation}
 {Our aim is to find a policy $R^* \in S_M(A_1)$ for 
each $M \in \mathbb{N}$ for which }
\begin{equation}
 W_M(A_1,H) \triangleq \inf_{R \in S_{M}(A_1)} W_M(A_1,R,H).
\end{equation}

 {We define for any $1 \leq j \leq M$, the average power spent
from decision 
time $j$ onwards as}
\begin{equation}
 \displaystyle{ W_{M-j}(q_j,h_j)=E_{H} \bigg[ \sum_{l=j}^M f(R_l,q_l,H_l) \bigg], \text {where } q_j =\sum_{l=j}^M R_l.}
\end{equation}
 {From Bellman's equation (\cite{lassi}, \cite{putterman}), we have $W_{M-j}(q_j,h_j)=$}
\begin{equation}\label{nine}
 \min_{R^{*}_j \in [0, q_j]} \left \{  f(R_j,q_j,h_j)+ E_{H_{j+1}}[W_{(M-j-1)}(q_{j+1},H_{j+1})]   \right \}.
\end{equation}
 
 { We obtain a closed-form solution using induction.
 Initially, we assume that the hard deadline constraint is one slot.
 Then, we need to transmit $A_1$ in the first slot itself.  Let $h_1$ be the channel gain
 in slot $1$. Then power consumption in the  first slot is $W_1(A_1,h_1)=\frac{e^{A_1}-1}{h_1}$. 
 Hence the average power consumption is}
 \begin{equation}
 E_{H}[W_{1}(A_1,H)]=   (e^{A_1}-1) E \left [ \frac{1}{H} \right ].          
\end{equation}
 Thus, (\ref{convexhardstar}) is satisfied for $M=1$.

Let for $M\geq 1$,
\begin{equation}
 E_H[W_M(A_1,H)]= M e^{\frac{A_1}{M}} \left( \displaystyle{\prod_{j=1}^M} E \left[ \frac{1}{H^{\frac{1}{j}}} \right]^{\frac{j}{M}} \right)-M E \left[ \frac{1}{H} \right].
\end{equation}

Now,we want to show that (\ref{convexhardstar}) also holds for $M+1$. From (\ref{nine}),
 \begin{align}
 W_{M+1}(A_1,h_1)= & \min_{R_1} \left \{ \frac{e^{R_1}-1}{h_1}+ E_{H}[W_M(A_1-R_1,H)]   \right \} \notag \\
=&\min_{0 \leq R_1 \leq A_1} \Biggl \{ \frac{e^{R_1}-1}{h_1}+ M e^{\frac{A_1-R_1}{M}}\times   \notag \\
&  \bigg( \displaystyle{\prod_{j=1}^M} E \left[ \frac{1}{H^{\frac{1}{j}}} \right]^{\frac{j}{M}} \bigg) -M E \left[ \frac{1}{H} \right] \Biggr \}. \label{indstar} 
\end{align}
By taking derivate w.r.t  $R_1$ and equating to zero. We get,
\begin{equation} \label{mstar}
 e^{R_1}= (h_1)^{\frac{M}{M+1}} e^{\frac{A_1}{M+1}} \left(\prod_{j=1}^M \left( E \left[ \frac{1}{H^{\frac{1}{j}}} \right] \right)^{\frac{j}{M+1}} \right).
\end{equation}

Substituting (\ref{mstar}) in (\ref{indstar}).  
\begin{eqnarray} \label{fifteen}
 W_{M+1}(A_1,h_1)&= &\frac{(M+1) e^{\frac{A_1}{M+1}} \left( \displaystyle{\prod_{j=1}^M} E \left[ \frac{1}{H^{\frac{1}{j}}} \right]^{\frac{j}{M+1}} \right)}{(h_1)^{\frac{1}{M+1}}} \notag \\
 &-&M E \left[ \frac{1}{H} \right] -\frac{1}{h_1}.
\end{eqnarray}
By taking expectation  of the above equation on both side. We get 
\begin{eqnarray*} \label{star}
 E_H[W_{M+1}(A_1,H)]= (M+1) e^{\frac{A_1}{M+1}} \left( \displaystyle{\prod_{j=1}^{M+1}} E \left[ \frac{1}{H^{\frac{1}{j}}} \right]^{\frac{j}{M+1}} \right)\notag \\
 -(M+1) E \left[ \frac{1}{H} \right]. 
\end{eqnarray*}

\end{proof}

{From Theorem 1,  if $A_1$ data is to be transmitted in a frame of size $M$, then in slot $k$, $1\leq k \leq M$, the data transmitted $R_k$
with remaining deadline of $M-k$ is
\begin{equation}
 \displaystyle{ R_{k}= \ln \left( (h_k)^{\frac{M-k-1}{M-k}} e^{\frac{A_1}{M-k}} \left(\prod_{j=1}^{M-k-1} \left( E \left[ \frac{1}{H^{\frac{1}{j}}} \right] \right)^{\frac{j}{M-k}} \right)\right)}.
\end{equation}}


 All our derivations
hold good for continuous channel case as well. We should replace the summations 
over the channel gains with the integrations over the channel distributions.

{From (\ref{convexhardstar}), we see that $W_M(A_1,h_1)$ is exponentially increasing in $A_1$. 
Also, since the policies for $M$ are a subset of policies for $M+1$, 
$E_H[W_{M+1}(A_1,H)] < E_H[W_{M}(A_1,H)]$. We can indeed prove that this inequality can be made strict.

 For the non-fading case, i.e., $H_k=1$. Hence, from (\ref{star})
 we get
  \begin{equation}
  R_k= \displaystyle{\frac{A_1}{M},\ \forall k,\ 1 \leq k \leq M.}
 \end{equation} }

We summarize this algorithm as  Algorithm 1 below.
We  initialize the
hard deadline constraint as $D_1=M$, $q_1=0$ and  place the data in the queue and queue evolves as
$q_k=q_k+A_k1_{\{(k-1) \mod M=0\}} $. It should be served within next
$M$ slots including the current slot of the arrival.
In slot $k$, we transmit
$R_{k}= \ln \left( (h_k)^{\frac{D_k-1}{D_k}} e^{\frac{q_k}{D_k}} \left(\prod_{j=1}^{D_k-1} \left( E \left[ \frac{1}{H^{\frac{1}{j}}} \right] \right)^{\frac{j}{D_k}} \right)\right)$
 and we also update the queue length as $q_{k+1} = \left( q_k-R_k \right)$. We update the remaining 
 hard deadline as $D_{k+1}=D_k-1+ M 1_{\{(k-1) \mod M=0\}}$.  We run this algorithm in 
 every slot, where $T$ is the total number of frames.
 
  \begin{algorithm}
  \begin{algorithmic}
  \caption{  {Average Power Optimal  for hard deadline,
  when the rate is concave function of  power}}
  \label{algo1}
   \STATE 1. Initialize $D_1=M $ and $q_1=0$
   \STATE 2. for $k=1:1:TM$ \\
    \STATE (i)  $q_k=q_k+A_k1_{\{(k-1) \mod M=0\}} $
    \STATE (ii) Let $h_k$ be the instantaneous  channel gain in time slot $k$. 
    \STATE (iii) Compute \\
   $ { R_{k}= \ln \left( (h_k)^{\frac{D_k-1}{D_k}} e^{\frac{q_k}{D_k}} \left(\prod_{j=1}^{D_k-1} \left( E \left[ \frac{1}{H^{\frac{1}{j}}} \right] \right)^{\frac{j}{D_k}} \right)\right)}$\\
     \STATE (iv)  Update $q_{k+1}=\left(q_k-R_k \right)$ 
      \STATE (v) Update $D_{k+1}=D_k-1+ M 1_{\{(k-1) \mod M=0\}}$
    \STATE endfor   
  \end{algorithmic}
 \end{algorithm}

{Now, we compare our optimal policy with the heuristic policies proposed in \cite{f005}.  
{ $A_k$
takes values from the set $\{ 0.5, 1, 1.5\}$ with equal probabilities.
Channel gains take values in the set $\{ 0.25$, $0.37$, $0.5$, $0.62\}$ with equal probabilities. 
Data comes only at  the
beginning of the frame  and should be served by the end of the frame. The duration 
of the frame is $M$ slots. 
We plot the optimal average power and the power consumed by the heuristic schemes in (\cite{wchen}, \cite{f005})
in  Fig. \ref{xyz44477788}. It is clear that the  average power consumption decreases with
the hard deadline $M$ and our scheme outperforms the other schemes in  (\cite{wchen}, \cite{f005}).

 \begin{figure}
\begin{flushleft}
\includegraphics[trim=3.9cm 9.7cm 4cm 10cm, clip=true, scale=0.70]{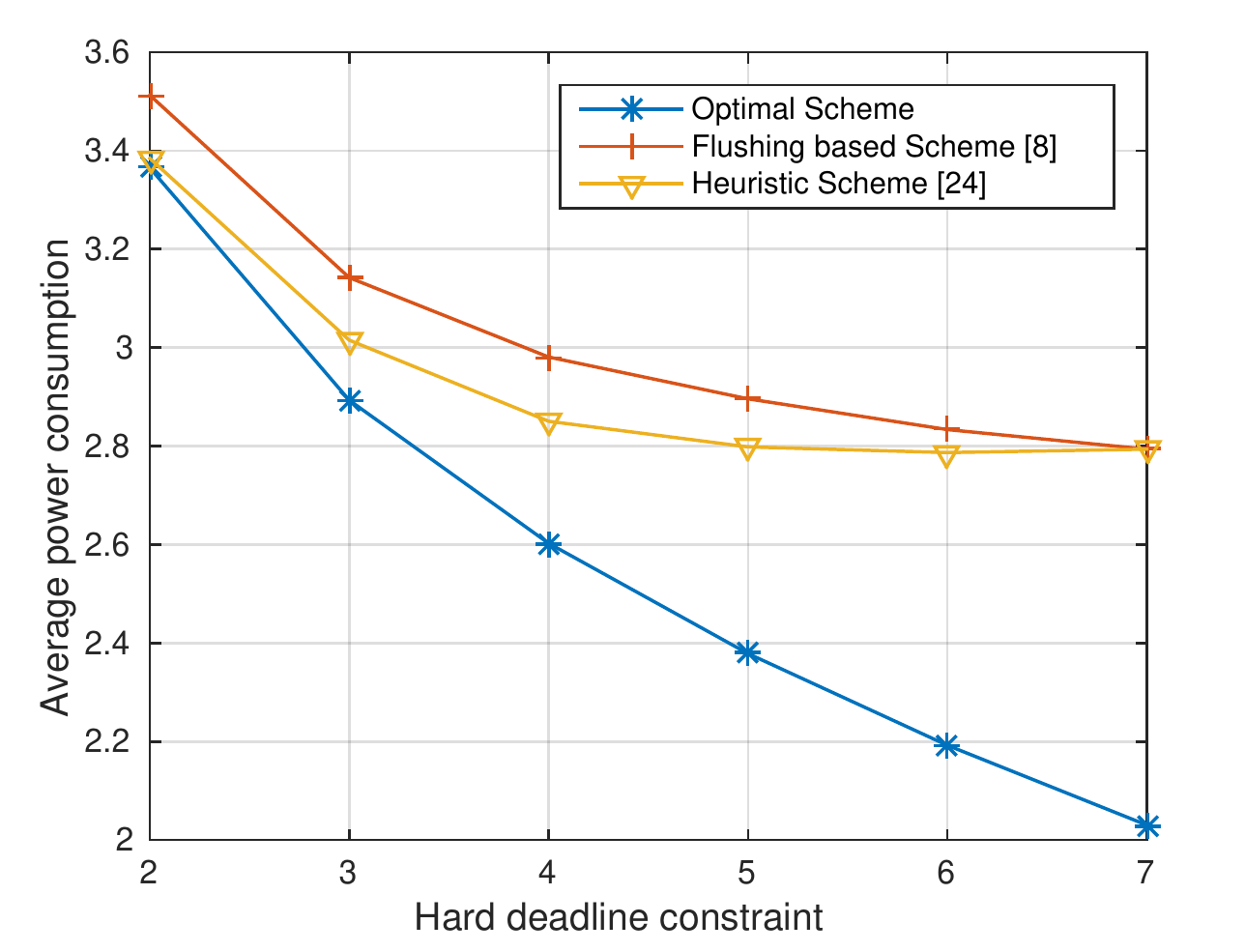} 
\caption{Average power consumption of optimal scheme and proposed heuristic schemes 
in \cite{f005} and \cite{wchen} vs hard deadline constraint $M$.}
\label{xyz44477788}
\end{flushleft}
\end{figure}



\section{Data Arrives  in Every Slot} \label{concave-case2}
 In this Section, we consider the case
 when the data arrives in  every slot of the frame and
 all the data which has arrived in the  frame
 should be served by the end of the frame. 
 We assume that $\{ A_k, k \geq 2\}$ is \emph{iid}.
 In the beginning of the frame we have data $\bar{A}_1$, independent of $\{ A_k, k \geq 2 \}$.
 Distribution of $\bar{A}_1$ can be different from that of $A_2$. 

The following theorem will provide us the algorithm. This is an extension of Theorem 1.
\begin{theorem}
Let the  data arrive in every slot of the frame and the 
  data arriving  between $[kM+1 \ (k+1)M]$, should be served by
$(k+1)M$. Also, let $\bar{A}_1$ nats be there in the beginning of slot $1$.
Then, the average power consumption for the optimal solution is, 
 \begin{eqnarray} \label{secondtheorem}
E_{{A_2},H}[W^{\prime}(\bar{A}_1,H)]= M e^{\frac{\bar{A}_1}{M}}  \left( \displaystyle{\prod_{j=1}^M} E \left[ \frac{1}{H^{\frac{1}{j}}} \right]^{\frac{j}{M}} \right)\notag \\
 \left( \prod_{j=1}^{M-1} E_{A_2} \left[  e^{\frac{A_2}{j}}\right]^{\frac{j}{M}} \right)-M E\left[ \frac{1}{H} \right],   
\end{eqnarray}
for all $M \geq 1$.
\end{theorem}

\begin{proof}
When the frame size is one, we have 
 \begin{equation}
 E_{H}[W^{\prime}_{1}(\bar{A}_1,H)]=   (e^{\bar{A}_1}-1) E \left [ \frac{1}{H} \right ].          
\end{equation}
Thus, (\ref{secondtheorem}) is satisfied. Let it be satisfied for $M \geq 1$. We show it for $M+1$. By (\ref{nine}), 
\begin{eqnarray} \label{mplusoneterm2}
 W^{\prime}_{M+1}(\bar{A}_1,h_1)&=& \min_{ 0 \leq R_1 \leq \bar{A}_1} \biggl\{ \frac{e^{R_1}-1}{h_1}  \\
 &+& E_{H}E_{A_2}[[ W^{\prime}_M(\bar{A}_1-R_1+A_2,H)]] \biggr \}. \notag
 \end{eqnarray}
By taking derivate w.r.t.  $R_1$ and equating to zero. We get 
\begin{eqnarray}\label{mplusonerateprime}
 e^{R_1}= (h_1)^{\frac{M}{M+1}} e^{\frac{\bar{A}_1}{M+1}}  \left( \prod_{j=1}^M  E\left[ \frac{1}{H^{\frac{1}{j}}} \right]^{\frac{j}{M+1}} \right) \times \\
 \left(\prod_{j=1}^M E_{A_2}\left[e^{\frac{A_2}{j}} \right]^{\frac{j}{M+1}} \right). \notag
\end{eqnarray}
Substituting  (\ref{mplusonerateprime}) in (\ref{mplusoneterm2}) and taking expectation over $H$, we get

\begin{align} 
&E_{A_2,H}[W^{\prime}_{M+1}(\bar{A}_1,H)]=(M+1) e^{\frac{\bar{A}_1}{M+1}} \times  \notag \\
&\left( \displaystyle{\prod_{j=1}^{M+1}} E \left[ \frac{1}{H^{\frac{1}{j}}} \right]^{\frac{j}{M+1}} \right) \left( \prod_{j=1}^{M} E_{A_2} \left[  e^{\frac{A_2}{j}}\right]^{\frac{j}{M+1}} \right) \notag \\
&-(M+1) E \left[ \frac{1}{H} \right].
\end{align}
\end{proof}

This gives the following Algorithm 2. 
We  initialize the hard deadline constraint as $D_1=M$, $q_1=0$  and set $q_k=q_k+A_k$, where $q_1=0$. 
In time slot $k$, we transmit 
\begin{eqnarray*}
 R_k= \ln \Biggl( e^{\frac{q_k}{D_k}} (h_k)^{\frac{D_k-1}{D_k}}  \bigg(\prod_{j=1}^{D_k-1} E \left[ \frac{1}{H^{\frac{1}{j}}} \right]^{\frac{j}{D_k}} \bigg) \times \notag \\
 \bigg( \prod_{j=1}^{D_k-1} E_{A_2} \left [ e^{\frac{A_2}{j}} \right]^{\frac{j}{D_k}} \bigg)\Biggr)
\end{eqnarray*}
 and we update the queue length as $q_{k+1} = (q_{k}- R_k)$ and $D_{k+1}=D_k-1$. 
  We run this algorithm  in  every frame.

 \begin{algorithm}
  \begin{algorithmic}
  \caption{  Data comes in whole frame}
  \label{algo1}  
   \STATE 1. Initialize $D_k=M $, $q_{1}=0$
   \STATE 2. for $k=1:1:TM$ \\
    \STATE (i) $q_{k}=q_k+A_k $\\
    (ii) Let $h_k$ is instantaneous channel gain in time slot $k$.  \\ \nonumber
    (iii)  Compute $R_k=$\\
       $\ln \left( e^{\frac{q_k}{D_k}} h_k^{\frac{D_k-1}{D_k}}  \Pi_{j=1}^{D_k-1} E \left[ \frac{1}{H^{\frac{1}{j}}} \right]^{\frac{j}{D_k}}    \left( \Pi_{j=1}^{D_k-1} E \left [ e^{\frac{A}{j}} \right]^{\frac{j}{D_k}} \right)\right)$\\
    (iv)   Update $q_{k+1}=(q_k-R_k)$ \\
    (v) Update $D_{k+1}=(D_k-1)+ M 1_{\{(k-1) \mod M=0\}}$\\
    \STATE endfor
  \end{algorithmic}
 \end{algorithm}


 \begin{figure}
\begin{flushleft}
\includegraphics[trim=3.9cm 9.7cm 4cm 10cm, clip=true, scale=0.70]{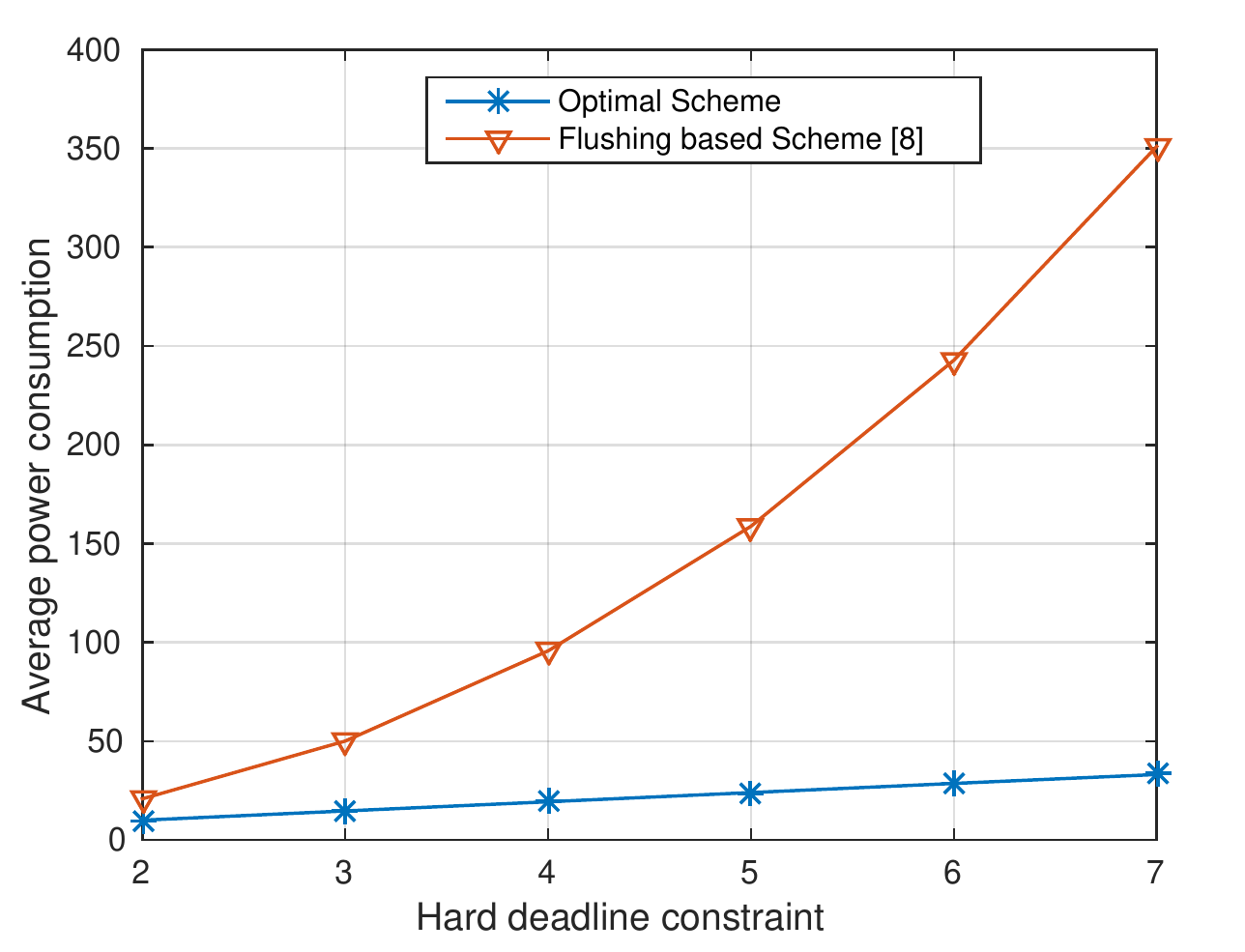}
\caption{Average power consumption vs hard deadline constraint.}
\label{xyz455777711}
\end{flushleft}
\end{figure}

From the expression (\ref{secondtheorem}) for $W^{\prime}_M(\bar{A}_1,h_1)$ we observe that the power 
consumed exponentially increases with $\bar{A}_1$.
But unlike, the results in Section III, $E_{A_2,H}[W^{\prime}_M(\bar{A}_1,H)]$
increases with $M$ if $\bar{A}_1$ and $A_2$ have the same distribution. This is
because there is traffic arriving during the frame itself. To see this,
let for the case of frame size, $M+1$, in slot $1$, we use power $P_1$
and transmit $R_1$ nats, then 
\begin{eqnarray*}
 &E_{{A}_2,H}[W^{\prime}_{M+1}(\bar{A}_1,H)]=P_1+ \notag \\
 &E_{{A}_2,H_1}[W_{M}^{\prime}(\bar{A}_1+A_2-R_1,H) ] 
 \geq E_{{A}_2,H}[W^{\prime}_M(\bar{A}_1,H)].
\end{eqnarray*}
If $\bar{A}_1$ and $A_2$ have different distributions then this inequality may not hold
because $\bar{A}_1+A_2-R_1$ may not be stochastically greater than $\bar{A}_1$.

Now, we compare our optimal policy with the algorithm proposed in \cite{wchen}.
 Arrival process
takes values from the set $\{ 2, 2.5, 3\}$ with equal probabilities. 
Channel gains take values in the set $\{ 0.5, 0.75, 1, 1.25\}$ 
with equal probabilities. We have run the algorithm and plotted the simulated curves for average 
 power consumptions versus the deadline constraint M in Fig. \ref{xyz455777711}. {Our 
 algorithms always outperforms the algorithm in \cite{wchen} and often substantially.}
\section{Single User, Multihop Network} \label{concave-case4}
We consider a network which is a connected, directed
graph $\mathcal{G(N,L)}$, where 
$\mathcal{N}$ 
is the set of nodes
and $\mathcal{L}$ 
is
the set of directed links. 
A subset of nodes (called source nodes) in the network
transmits data to another subset of nodes (called
destination nodes).
Each source has one destination. The time axis is slotted.
 The stream of packets transmitted from a source node to its
 respective destination node is called a \emph{flow}.
 
The set of user flows $\{1,2,\ldots,N\}$ is denoted by $\mathcal{F}$. 
Let $A^f_k$ be the number of nats generated by flow 
$f$ in slot $t$ at its source.
We assume $\{A^f_k,\ k\geq 0\}$ to be \emph{iid},
independent for different flows. We will also assume
$E[A^f]<\infty$ for each $f$. End-to-end hard deadline 
for flow $f$ is denoted as $D^f$.

We assume that  when a node is transmitting to some other node, 
it cannot receive data from any other
node(s). Similarly, when a node is receiving data from a  node, it
cannot transmit data to any other
node. Other transmission constraints can be included in our setup (also, we can modify 
the setup to we allow fewer constraints, e.g., we allow full duplex links). 

Because of these constraints we can divide the set of links $\mathcal{L}$ into
independent sets. All links in a set can transmit at the same time, causing
negligible interference to each other. However, the links in two different
sets cannot transmit in the same slot. Efficient algorithms to obtain independent
sets in a graph are available in \cite{robo}.

Let $(i,j)$ be the link which connects node $i$ to node $j$.
 Let the channel gain in slot $k$
for link $(i,j)$ be $H_{ij}^k$, which is available 
to the nodes $i$ and  $j$ at the beginning of slot $k$.
We assume that the channel gain $H_{ij}^k$ remains constant during  slot $k$.  
We also assume that the channel gain  process $\{H_{ij}^k,\ k \geq 0\}$ 
is \emph{iid} on all links
and independent for different links.
The channel gain $H_{ij}^k$ takes values on a finite set. 
If in time slot $k$, the power spent by node $i$ for flow $f$
is $P_{ij}^k(f)$ then the data $R_{ij}^k(f)$ transmitted to node $j$ for flow
$f$ is,
\begin{equation}
 R_{ij}^k(f)=\frac{1}{2} \ln(1+ \gamma_{ij} P_{ij}^k(f) H_{ij}^k /\sigma^2_{ij}), \label{first8}
\end{equation}
 where $\sigma^2_{ij}$ is the receiver noise variance and $\gamma_{ij}$ is a constant that depends 
 on the modulation and coding used. For simplicity, we assume that $\sigma^2_{ij}=1$, 
and $ \gamma_{ij}=1, \forall i,j$.

In this section, we consider the problem of a single user. 
 For simplicity, from now onwards, in this section we will remove $f$ from the notation.
Our objective is to minimize the overall average power consumption of the 
wireless network subject to the end-to-end deadline for all packets.
We will use Algorithms 1 and 2, 
to solve the single user problem. 

Suppose for the flow, the overall end-to-end hard deadline is $D$ slots.
If we selected a path $\mathcal{P}$ for this flow with links $l_1, l_2, \cdots, l_{N_1}$, 
we can split the hard deadline $D$ into deadlines $D_i, i=1,\cdots, N_1$ for each of the
links such that $D_i\geq 1$ and $\sum_{i=1}^{N_1}D_i \le D$. We should choose $D_i$'s such that
the overall power required $\sum_{i=1}^N P_i$ on these links is minimized where $P_i$ is 
the power spent on link $l_i$. The power $P_i$ required on these links is obtained from 
Theorem 2 (for the source node 1) and Theorem 1 for the other nodes.

The link scheduling is done via TDMA on the independent sets.
If there are $N_2$ independent sets in the graph, we retain only the independent 
sets in which at least one of the links on $\mathcal{P}$ resides (for multiuser scenario,
we may have to consider all independent sets). For simplicity, we assume that link $i$ is 
in set $S_{i}$. We assign $D_i$ consecutive slots to set $S_i$. 
Also we assume that these links get their slots one after another on $\mathcal{P}$. 
Furthermore, we also assume that a link is in only one of the independent sets.
Our procedure can be adapted to the general case. However a set $S_i$ will often have multiple links. 
Thus, we will consider the case where we take a route $\mathcal{P}$ where link $i$ has deadline
$D_i$ and we assign $D_i$ consecutive slots to the set $S_{i}$. We will not 
insist that $\sum_{j=1}^{N_i}D_j = D$. We will in fact see that often $\sum_{j=1}^{N_i}D_j < D$ will provide
less power, something which seems counterintuitive. Let there be $N_3$ independent sets needed for the links $l_1,l_2,\cdots, l_{N_1}$.

In this setup, node 1 will transmit for $D_1$ consecutive slots 
and then wait for $D2+D_3+\cdots+D_{N_3}$ slots for other independent
sets to transmit and then again transmit for $D_1$ slots and so on.
During these $D_1$ slots, it will get $A_1,\cdots, A_{D_1}$ \emph{iid} nats 
for transmission. Also it will have $\sum_{k=1}^{D_2+\cdots+D_{N_3}}A_k$ nats 
in the beginning of each of its frames (of size $D_1$ slots)
generated  by the source when it is not transmitting, 
for transmission. 
Thus by Theorem 2, the energy needed by node 1 to transmit all this data in a frame of size $D_1$
is  given 
\begin{align}
&= E_{A_1,H}\left[W^{\prime}_{D_1}\left(\sum_{k=1}^{D_2+\cdots+D_{N_3}}A_k, H\right)\right] \label{eq1}
\end{align}

\begin{align*}
&=D_1 E[e^{\frac{A_1}{D_1}}]^{D_2+\cdots+D_{N_3}}\biggl(a_{D_1}\biggr) \left(\prod_{j=1}^{D_1-1}E[e^{\frac{A_1}{j}}]^{\frac{j}{D_1}} \right) \notag \\
&-D_1E\left[\frac{1}{H_1}\right],
\end{align*}
where $a_{D_i} = \prod_{z=1}^{D_i} E\left[ \frac{1}{H_i^{\frac{1}{z}}} \right]^{\frac{z}{D_i}}. $

For node 2 on the path, all the data it will transmit in its frame comes at the beginning of its frame of size $D_2$ slots. The data arriving is $\sum_{k=1}^{D_1+\cdots+D_{N_3}}A_k$. Hence, by Theorem 1, the energy required to transmit it is 
\begin{align}
&E_{H}\left[W_{D_2}\left(\sum_{k=1}^{D_1+\cdots+D_{N_3}} A_k, H \right)\right] \notag \\
&= D_2E\left[e^{(\sum_{k=1}^{D_1+\cdots+D_{N_3}}A_k)/D_2}\right]\biggl(a_{D_2}\biggr) -D_2E\left[\frac{1}{H_2}\right] 
\label{eq2}
\end{align}
Similarly for node $i$ on the path $\mathcal{P}$, we have $D_i$ slots in its frame and data $\sum_{k=1}^{D_{N_3}}A_k$ arrives only in the beginning of its frame. Thus the energy used is 
\begin{align}
&E_H \left[W_{D_i}\left(\sum_{k=1}^{D_1+\cdots+D_{N_3}}A_k, H \right)\right] \notag \\
&= D_iE\left[e^{(\sum_{k=1}^{D_1+\cdots+D_{N_3}}A_k)/D_i}\right]\biggl(a_{D_i}\biggr) -D_iE\left[\frac{1}{H_i}\right] 
\label{eq3}
\end{align}

Now we look at the power (\ref{eq1}) and (\ref{eq3}) to see what
could possibly be good deadlines $D_i$. In (\ref{eq1}), we see as $D_1$
increases, keeping $D_2,\cdots,D_{N_3}$ same, $a_{D_1}$ decreases;
$-D_1E\left[\frac{1}{H_1}\right]$ will also reduce power. But in the
first term, some of the subterms increase. However, we will see that 
the dominant term is $E[e^{\frac{A_1}{D_1}}]^{D_2+\cdots+D_{N_3}}$. 
This exponentially decreases in $D_1$ and its effect is amplified by
exponent $D_2+\cdots+D_{N_3}$. Thus as $D_1$ increases while other 
$D_i$s stay constant, we expect that power spent at node 1 will decrease. 
On the other hand if any of the other $D_i$'s increase, its power requirement will exponentially increase.

Now we look at (\ref{eq3}) for node $i$. If $D_i$ increases but 
other $D_j$ are fixed, then $a_{D_i}$ decreases. Also $E[e^{A_1/D_i}]$
decreases and its effect is amplified by power $D_i$. Thus, we expect its
power to decrease. But if $D_i$ is fixed and any other $D_j$ increases, 
this exponentially increases the power required at $D_i$.

Thus, we conclude that if we increase the deadline of a node, 
its power requirement decreases but it will increase exponentially the power requirement of all other nodes.

Thus if path length $|\mathcal{P}|>2$, we should set the deadline of each link
as 1 even if $|\mathcal{P}|<D$. 

We illustrate the above procedure by considering the important 
special case where the only constraint on transmission is that a 
node can only transmit or receive and that too on only one of the 
links at a time.Then for any path $\mathcal{P}$, we have two independent
sets $S_1$ and $S_2$. We can define $S_1=\{1,3,5,\cdots\}$ and
$S_2=\{2,4,6,\cdots\}$ on our path. As argued above, we take $D_1=D_2=1$.
Let us assign slots $[k,k+1)$ to $S_1$ if $k$ is odd and $S_2$ otherwise. 
Then power needed by node $i$ on the path is
\begin{equation}
E[e^{A_1+A_2}-1]E\left[\frac{1}{H_i}\right] = (E[e^{A_1}]^2-1)E\left[\frac{1}{H_i}\right] 
\label{eq4}
\end{equation}
for all $i$.

If $|\mathcal{P}|=1$ and we take $D_1=1$, then the power needed is 
\begin{equation}
(E[e^A]-1)E\left[\frac{1}{H_1}\right]. 
\label{eq5}
\end{equation}

Thus to minimize power, we should choose one hop path over $\mathcal{P}$ if 
\begin{equation}
(E[e^A]-1)E\left[\frac{1}{H_1}\right] \leq (E[e^{A_1}]^2-1)\sum_{i\in \mathcal{P}}E\left[\frac{1}{H_i}\right]
\label{eq6}
\end{equation}
and hence if
\begin{equation}
E\left[\frac{1}{H_1}\right] \leq (E[e^{A_1}]+1)\sum_{i\in \mathcal{P}}E\left[\frac{1}{H_i}\right].
\label{eq7}
\end{equation}
Often it will be true unless the traffic is very low. If this condition is violated, then we should look for a path with minimum $\sum_{i\in \mathcal{P}}E\left[\frac{1}{H_i}\right]$. This can be computed via Dijkstra's algorithm by keeping the cost of each link with channel gain $H$ as $E\left[\frac{1}{H}\right]$.

We illustrate this special case with an example.
Let us take the end-to-end deadline $D=10$ slots.
We consider a network with $15$ nodes.  Source node is $1$ 
and the destination node is $9$. 
 To form the graph, we generated a random binary  matrix $\textbf{C}=[c_{ij}]$ of size $ 15 \times 15$.
 Its element $c_{ij}$  is $1$ if there is a link from  node $i$ to  
 node $j$; otherwise $c_{ij}=0$. 
We compute the optimal path from the source node to the destination
node using Dijkstra algorithm by taking weight on link $(i,j)$ to be
$E\left[\frac{1}{H_{ij}}\right]$. 
Then, for our channel gain distributions (not provided here to conserve space), the optimal 
route $\mathcal{P}$ obtained is $ 1 \rightarrow 5 \rightarrow 7 \rightarrow 9$.
Channel gains on links $(1,5)$, $(5,7)$ and $(7,9)$ are 
$\{2$, $3$, $4$, $5\}$, $\{0.2$, $0.5$, $0.8$, $1\}$,
$\{2$, $2.5$, $2.9$, $3.5\}$ and occur with equal probability. The arrival
process at the  source node takes values from the set $\{1,2,3\}$ with equal probabilities.

{    
The number
of independent sets, taking only the constraints that a node can either 
transmit or receive on a single link in a slot, are two: Set $S_1=\{(1,5), (7,9)\}$ and $S_2=\{(5,7)\}$.
Then end-to-end hard deadline of $D$ is split into the deadlines $D_{i}$ on these links with
$D_{i} \geq 1$ and $\sum_{i\in \mathcal{P}} D_{i} \leq D$.
 }


Since $D_{i} \geq 1$, for this path, the end-to-end deadline $D$ for packets needs to be  $ \geq 4$ (we will see it below).
If this condition is not met 
then we can compute the 
second least cost path (say via \cite{kpath}) and keep finding successive least cost paths till we get one with hop count $\leq D-1$
{(for $D>1$; for $D = 1$, we need a direct link from source to destination).}
If no such path exists then the deadline $D$ is not feasible in this network.

By taking $D_{1}=3$, $D_{2}=4$, we get $D_{1}+D_{2}+D_{3}=10$
slots. 
If we insist on keeping end-to-end deadline 10, we can show that this is the optimal breakup of the deadlines. For this the 
theoretical and simulated average power on links $(1,5)$, $(5,7)$ and $(7,9)$
are $(128.2,129.5)$, $(310.8,312.2)$ and $(155.3,156.1)$ respectively.

But now we take $D_{1}=D_{2}=1$. Then, the theoretical and simulated average power on links $1$, $2$ and $3$ 
are $(32.1,32)$, $(231.9,229.6)$ and $(38.5, 38.3)$ respectively.
 This is far lower than if we insists on a deadline
of $10$. 

Now we illustrate some other properties mentioned above based on which we concluded that $D_{i}=1$
is the best 
for the overall sum of the power. In Table I, we provide the powers needed at node $1$, $5$ and $7$ when  $D_{1}=1$ and 
$D_{2}$ is increased. As we claimed, the power at node $5$ decreases but power at nodes $1$ and $7$ increases exponentially. 
The total power also increases substantially. We have seen a similar effect if $D_{2}$ is fixed $=1$ and
$D_{1}$ is increased. Then (see Table II) the power at node $1$ and $7$ decreases but at node $5$ increases exponentially. 
The total end-to-end power increases drastically.  This verifies the claims we made above.

\begin{table}[h]  \label{table61}
\centering
\caption{{Single user, Multihop: $D_1=1$, powers at different links}}
  \resizebox{\columnwidth}{!}{ 
    \begin{tabular}{ | c | c | c | c |c| }
    \hline
    &$L_{1,5}$ & $L_{5,7}$ & $L_{7,9}$  \\ \hline
    $D_2=1$ & $ 32$ & $ 231.9 $ & $ 38.3 $  \\ 
    $D_2=2$ & $ 329 $  & $ 33.5 $ & $ 389 $  \\
    $D_2=3$ & $ 3.2 \times 10^3 $ & $ 19.1$ & $ 3.9 \times 10^3 $  \\
    $D_2=4$ & $ 3.3 \times 10^4 $ & $ 14.4 $  & $ 3.95 \times 10^4 $  \\
     $D_2=5$ & $ 3.5 \times 10^5 $ & $ 12 $  & $ 3.97 \times 10^5$  \\
    \hline
    \end{tabular}
    \vspace{0.05cm}
    
    }
\end{table}

\begin{table}[h]  \label{table61}
\centering
\caption{{Single user, Multihop: $D_2=1$, powers at different links}}
  \resizebox{\columnwidth}{!}{ 
    \begin{tabular}{ | c | c | c | c |c| }
    \hline
    &$L_{1,5}$ & $L_{5,7}$ & $L_{7,9}$  \\ \hline
    $D_1=1$ & $ 32$ & $ 231$ & $ 38.3 $  \\ 
    $D_1=2$ & $ 17 $  & $ 2.33 \times 10^3 $ & $ 18.76 $  \\
    $D_1=3$ & $ 15.9 $ & $ 2.34 \times 10^4 $ & $ 17.94 $  \\
    $D_1=4$ & $ 16.9 $ & $ 2.38 \times 10^5 $  & $ 17.82 $  \\
     $D_1=5$ & $ 18.4 $ & $ 2.4 \times 10^6 $  & $ 17.6 $  \\
    \hline
    \end{tabular}
    \vspace{0.05cm}
    
    }
\end{table}

\section{Multiuser, Multihop Network} \label{concave-case5}
 In this section, we consider the case, when
there are multiple source-destination pairs and the data arrives
 in every slot at the source nodes and it has the same hard deadline for every packet
 of a given flow $f$. We  demonstrate our algorithm via an example.
 
 Consider the example
 of $15$ nodes of Section V. There are two source-destination pairs $(1,8)$ and $(2,6)$. 
 We have no direct links between source-destination pair.
 Thus, we compute the optimal paths for each source-destination pair via Dijkstra's algorithm
 where the cost of each link is $E\left[\frac{1}{H_{ij}}\right]$. The optimal path for source 1 is
 $1 \rightarrow 4 \rightarrow 5 \rightarrow 7 \rightarrow 8 $ and for source 2 is $2 \rightarrow 4 \rightarrow 5 \rightarrow 6.$

 To compute the energy, we also need to know $D_{i,j}$ (deadline corresponds to node $i$ to node $j$) the deadline we fix for each link.
 Based on single user results, we take $D_{i,j}=1$ for all links.
 
 Right now we just keep the two paths $1\to 4\to 5\to 7\to 8$ and $2 \to 4\to 5 \to 6$ to explain the rest of the algorithm.

 Let the  end-to-end hard deadline delay for flows $1$ and $2$ be $14$ and $10$  respectively.
  {  Channel gains on links $(1,4)$, $(4,5)$, $(5,7)$, $(7,8)$, $(2,4)$, $(5,6)$
 take value from sets $\{0.8$, $1.6$, $2.4$, $3.2$, $4  \}$, 
 $\{0.6$, $1.2$, $1.8$, $2.4$, $3\}$,
 $\{0.7$, $1.4$, $2.1$, $2.8$, $3.5 \}$,
 $\{0.9$, $1.8$, $2.7$, $3.6$, $4.5\}$,
 $\{1$, $2$, $3$, $4$, $5 \}$, $\{0.8$, $1.6$, $2.4$, $3.2$, $4 \}$.}
 
{Arrivals  for the both flows take  value from the set $\{ 1, 2, 3\}$ nats with equal probabilities
 and the end-to-end hard deadline for both flows is 10. }
 
{Independent sets for this example are $S_1=\{(1,4)$, $(5,6)$,  $(7,8) \}$,
 $S_2=\{ (2,4)$, $(5,7)\}$ and $S_3=\{(4,5)$, $(7,8)\}$. 
 First, we allocate
 one time slot per set  as follows:
 $S_1$, $S_2$, $S_3$, $S_1$, $S_2$, $S_3$, ....
 Each node transmits all its data in the slot alloted to it. 
 Then  flow $1$ experiences a maximum hard deadline delay  of $8$ slots and flow $2$ 
 experiences a maximum hard deadline delay of $5$ slots. Thus both the deadlines are met and we are done. }


{ The simulated and theoretical average power consumption for the above schemes for link
 $(1,4)$ is $(587$, $581)$, for
 link $(4,5)$ is $(7.9 \times 10^5$, $7.91 \times 10^5 )$, for link $(5,7)$ is $(658$, $664)$,
 link $(7,8)$ is $(519$, $516)$, link $(2,4)$ is $(469$, $465)$, and 
 link $(5,6)$ is $(576$, $581 )$.} 

 We can improve over this path selection by noticing that link $4 \to 5$ is common on the two paths.
 This increases the cost of this link. This can be taken into account as follows.
 
 To compute the optimal path for source $2$, we keep the cost of the links which are 
 not on this path as $E\left[\frac{1}{H_{ij}}\right]$. But the cost of any of the
 links on this path is obtained by considering the increase in energy needed on these links in
 transmitting data of source $2$ in addition to that of source $1$ (e.g., compute 
 the mean energy needed on link $(4,5)$ via Theorem $1$ when both the sources send
 data through it minus the energy when only source $1$ transmits). Of course we could have 
 taken the reverse order of first selecting route for source $2$ and then source $1$.
 Then we should keep the  paths for the two flows which provide the lowest sum energy.

{ Finally we have the following steps for our algorithm. First we find the routes for the different
flows via Dijkstra's algorithm as explained above, including the incremental cost mentioned above. In the next 
 step, we allocate one slot transmission time duration 
 for each set in round robin fashion.
 If we  meet the hard deadline constraints
 for all flows we are done; otherwise we find next best route for the flow(s), 
 till all flows meet their end-to-end hard deadline constraints.}
 
{One method to reduce the overall average power consumption is to reduce the load on the link which consumes more
 average power consumption. 
 This can be done by removing some flows on this link.
 One can also reduce the overall average power consumption of the system by reducing the number 
 of independent sets for the system. }

 




\section{Conclusions} \label{s4_conclusions}

We have  considered the 
problem of minimizing the average power
in the presence of hard deadline constraints. We consider the case
when the rate satisfies the generalized Shannon's formula. We have obtained  closed-form
optimal solutions when the data comes at the beginning of the frame
and should be served by the end of the frame. We have also obtained
closed-form optimal solutions when the data comes in 
every slot of the frame and should be served within the frame.
We have extended our
single user, single
server results to a single/multi user, multihop network when every user has its own end-to-end 
hard deadline constraint.

\end{document}